\documentclass[preliminary,publicdomain]{eptcs}
\providecommand{\event}{ACT 2020} 
\usepackage{underscore}           

\usepackage[utf8]{inputenc}
\usepackage[english]{babel}
\usepackage[T1]{fontenc}
\usepackage{amsmath,amsfonts,amssymb,amsthm, mathtools}
\usepackage{enumerate}
\usepackage{tikz-cd}
\usepackage{hyperref}
\usepackage[capitalise]{cleveref}
\usepackage{mathpartir}
\usepackage{microtype}

\usepackage{macros}

\newtheorem{theorem}{Theorem}
\newtheorem{lemma}[theorem]{Lemma}

\newtheorem{proposition}[theorem]{Proposition}
\theoremstyle{definition}
\newtheorem{definition}[theorem]{Definition}
\newtheorem{example}[theorem]{Example}
\theoremstyle{remark}

\title{Compositional Game Theory, Compositionally}
\author{\setlength{\tabcolsep}{15pt}
	\begin{tabular}{ccc}
		Robert Atkey & Bruno Gavranovi{\'c} & Neil Ghani\\
		Clemens Kupke & J{\'e}r{\'e}my Ledent & Fredrik Nordvall Forsberg
	\end{tabular}
	\institute{The MSP Group, University of
		Strathclyde\\ Glasgow, \'{E}cosse Libre}
}
\def\titlerunning{Compositional Game Theory, Compositionally}
\def\authorrunning{R. Atkey, B. Gavranovi{\'c}, N. Ghani, C. Kupke, J. Ledent \& F. Nordvall Forsberg}

\begin{document}
\maketitle

\begin{abstract}
 We present a new compositional approach to compositional game theory (CGT)
 based upon Arrows, a concept originally from functional programming, closely related to Tambara modules, and operators to build new Arrows from old. We
 model equilibria as a
module over an Arrow and define an operator to build a new Arrow from
such a module over an existing Arrow. We also model
 strategies as graded Arrows and define an operator which builds a
 new Arrow by taking the colimit of a graded Arrow. A final operator
builds a graded Arrow from a graded bimodule. We use this
 compositional approach to CGT to show how known and
 previously unknown variants of open games can be proven to form symmetric monoidal
 categories.
\end{abstract}

\section{Introduction}

A new strain of game theory --- {\em Compositional Game Theory (CGT)}
--- was introduced recently~\cite{GhaniHWZ18}.  At its core, CGT
involves a completely new representation of games --- {\em open games}
--- with operators for constructing larger and more complex games from
smaller, simpler (and hence easier to reason about) ones. Despite
recent substantial interest and further
development~\cite{julesMorphisms,GhaniKLF18,backprop,dioptics,BoltHZ19,probOG,JulesGameSemantics},
open games remain complex structures, e.g.\ the simplest form of an
open game is an 8-tuple consisting of 4 ports, a set of strategies,
play and coplay functions and an equilibrium predicate. More advanced
versions~\cite{BoltHZ19} require sophisticated structures such as
coends. This causes problems: (i) complex definitions lead to complex
and even error prone proofs; (ii) proofs are programs and overly
complex proofs turn into overly complex programs; (iii) even worse,
this complexity deters experimentation and innovation, for instance the
authors of~\cite{GhaniKLF18} were deterred from trying alternative
definitions as the work became prohibitive; and (iv) worse still, this
complexity suggests we do not fully understand the mathematical
structure of open games.

Category theorists have an answer to such complexity: trade complex
definitions in a simple category for simpler definitions in a more
complex category. This simplified open games by factoring their
definition via the category of lenses, and this more abstract
perspective led directly to the discovery of Bayesian Open
Games~\cite{BoltHZ19} where one naturally swaps out lenses in favour
of optics when one starts with a monoidal category. Swapping out a
subcomponent for another is a hallmark of compositionality, but,
because the definition of open games is not compositional, one has to
redevelop all the structure of open games built over optics from
scratch. The same thing happens with open games with mixed
strategies~\cite{GhaniKLF18} where a small change --- here to the type of the
equilibrium predicate --- leads to the need to redevelop the
full theory of these open games. The same story is true if we want to
(i) replace equilibria with best-response; or (ii) add type dependency;
or (iii) allow the covariant play function and contravariant co-play
function to live in different categories; or (iv) build open games over
monadic effects; or (v) consider a host of other variations
of open games. Every single time we must show from scratch that a
given variant of open game forms a monoidal category --- a highly
non-trivial task.

We present a compositional approach to CGT which allows us exactly to
make small modifications to the definition of an open game, check
these modifications have suitable local properties, and conclude the
resulting open games do indeed form a monoidal category. As with
compositionality in general, we start with a semantic structure and
define operators to build more complex examples from smaller ones.
The semantic structure we use is a presentation of monoidal
categories using profunctors and Arrows~\cite{hughesArrows}. This simplifies the
treatment of the monoidal structure to that of a strength, and also
allows us to work in a purely functorial setting. We stress that
arrows are just a convenient tool --- it would also be possible to work with monoidal categories directly.  This paper develops a three phase
compositional axiomatisation of CGT: (i) the first phase consists of
selecting an arrow to represent the play/coplay structure of open
games; (ii) the second phase
abstracts the equilibrium predicate to a bimodule over an arrow, and
defines an operator for building new arrows from an old arrow and a
bimodule over it; and (iii) the final phase abstracts the indexing of
open games by strategies by defining an operator which maps a graded
arrow to an arrow.  This provides a simple, compositional
axiomatisation of CGT in arrow theoretic terms.  We also show we can
add the strategies first and then later add the equilibrium
predicate. This later approach also allows us to treat probabilistic games
and open games based upon best response functions. The closest
approach in the literature is Hedges~\cite[\S 8]{JulesGameSemantics}, and while there are notable
similarities there are also significant differences, e.g.\ while Hedges~\cite{JulesGameSemantics} provides an
axiomatisation of certain forms of open games, it does not do so
compositionally, and is not general enough to cover probabilistic or
dependently typed open games.

In summary, this more abstract picture clearly exposes the
fundamental mathematical structure of open games in a way which has not
been done before. The originality of the paper lies in its conceptual
reimagining of open games as a compositional structure using canonical
mathematical structures such as arrows, modules and grading. Its
utility is demonstrated by the variety of examples --- some known and
others new to this paper --- encompassed by our framework. As a result,
we believe our presentation will both (i) lead to further new variants
of open games; and (ii) become the default presentation of open games.


\section{Profunctors and Arrows}
\begin{definition}[Profunctor]
	Given two categories $\cC$ and $\cD$, a \emph{profunctor}~$F$ from~$\cC$ to~$\cD$, denoted by $F : \cC \profto \cD$, is a functor $F : \cC^\op \times \cD \to \Set$.
\end{definition}

\begin{example} 
Let $\cC$ be a category. Then $\Hom{\cC}: \cC^\op \times \cC \to
\Set$ is a profunctor.
\end{example}

Profunctors, strong profunctors and monoids therein form our bedrock
and much of this paper is a toolbox for building new examples of
arrows from old. The bicategory $\Prof$ has categories as objects,
profunctors as morphisms, and natural transformations as $2$-cells.
Composition of profunctors is given by the the following coend
formula, for $F : \cC \profto \cD$ and $G : \cD \profto \cE$.
\[ (G \bullet F) (X,Z) := \int^{Y:\cD} F(X,Y) \times G(Y,Z) \]
The identity profunctor on~$\cC$ is $\Hom{\cC}$. We now fix a monoidal category~$\cC$ and focus on the category~$\Prof(\cC)$ of profunctors from~$\cC$ to~$\cC$ and natural transformations between them.

\begin{definition}[Strength]
\label{def:str}
	A \emph{(right) strength} for a profunctor $F : \Prof(\cC)$ is
	a family of morphisms
	$\st_{X,Y,Z} : F(X,Y) \to F(X \otimes Z, Y \otimes Z)$
        (natural in~$X,Y$, and dinatural in~$Z$) making the following two diagrams
	commute.
	We often write~$\st_Z$ or just~$\st$ when the subscripts can be inferred from context.
	\begin{center}
		\begin{tikzcd}
			F(X,Y) \arrow[d,"\st_Z",swap] \arrow[r,"\st_{Z\otimes Z'}"]
			& F(X \otimes (Z \otimes Z'), Y \otimes (Z \otimes Z')) \arrow[d,"{F(\alpha^{-1},\alpha)}"]
			\\
			F(X \otimes Z, Y \otimes Z) \arrow[r,"\st_{Z'}"]
			& F((X \otimes Z) \otimes Z', (Y \otimes Z) \otimes Z')
		\end{tikzcd}
	\qquad and \qquad
		\begin{tikzcd}
			F(X,Y) \arrow[d,"\st_I",bend left] \arrow[d,"{F(\rho,\rho^{-1})}",bend right,swap] \\
			F(X \otimes I, Y \otimes I)
		\end{tikzcd}
	\end{center}
\end{definition}

Such a pair $(F,\st)$ is called a \emph{strong
  profunctor}\footnote{See also Rivas and Jaskelioff~\cite{RivasJ17} for a similar
  definition in the case when $\cC$ is cartesian. This definition is
  also the same as Pastro and Street's definition of the additional
  structure required for a profunctor to be a \emph{Tambara
    module}~\cite{pastroStreet2008}. }. Note that several strengths may
exist for the same profunctor.
When several strengths for different profunctors (say $F$ and~$G$) are involved, we distinguish them using superscripts, $\st^F$ and~$\st^G$. A \emph{strong natural transformation}
from $(F,\st^F)$ to $(G,\st^G)$ is a natural
transformation~$\tau : F \to G$ that commutes with the strengths:
\begin{center}
	\begin{tikzcd}
		F(X,Y) \arrow[rr, "\st^F"]\arrow[d,"\tau_{X, Y}",swap]
		&& F(X \otimes Z, Y \otimes Z) \arrow[d,"\tau_{X \otimes Z, Y \otimes Z}"] \\
		G(X,Y) \arrow[rr,"\st^G"]
		&& G(X \otimes Z, Y \otimes Z)
	\end{tikzcd}
\end{center}

The category $\StrProf(\cC)$ of strong profunctors and strong natural transformations is monoidal, where the tensor product is the composition operator~$\bullet$,
and the unit is the functor $\Hom{\cC}$. This allows us to talk about
monoids inside $\StrProf(\cC)$.

\begin{definition}[Arrow]
\label{def:arrow}
An \emph{arrow} is a monoid in $(\StrProf(\cC), \bullet, \Hom{\cC})$,
i.e.\ a strong profunctor $(A,\st)$ together with strong natural transformations
${\pure : \Hom{\cC} \to A}$ and $\comp : A \bullet A \to A$, satisfying
the expected monoid laws.
\end{definition}


\noindent
If we unpack \cref{def:arrow} (and use the universal property of coends), an arrow on~$\cC$ can be described more concretely by the following data:
\begin{itemize}
	\item A family of sets $A : |\cC| \times |\cC| \to \Set$;
	\item A family of functions $\pure_{XY} : \Hom{\cC}(X,Y) \to A(X,Y)$;
	\item A family of functions $\comp_{XYZ} : A(X,Y) \times A(Y,Z) \to A(X,Z)$, called \emph{arrow composition};
	\item A family of functions $\st_{XYZ} : A(X,Y) \to A(X \otimes Z, Y \otimes Z)$;
\end{itemize}
subject to a number of \emph{arrow equations} expressing functoriality, naturality, and so on (see e.g.\ Atkey~\cite{bobArrows}).
We often use this presentation in the rest of the paper. We sometimes omit the subscripts to reduce notation when they can be inferred; and sometimes add a superscript, e.g.\  $\comp^A$, to specify which arrow they belong to.

\medskip

There is a notion of a \emph{left} strength $F(X,Y) \to F(Z \otimes X, Z \otimes Y)$, analogous to \cref{def:str}, that satisfies similar properties.
When the category $\cC$ is symmetric monoidal, left and
right strengths are interdefinable:
\begin{lemma}
	Let $\cC$ be symmetric monoidal. A right strength
	$\st$ for a profunctor $F : \Prof(\cC)$, induces a left
	strength $\st' : F(X,Y) \to F(Z \otimes X, Z \otimes Y)$ as
	$\st'_{XYZ} = F(\sigma_{ZX}, \sigma_{YZ}) \circ \st$. The left
	strength satisfies the (di)-naturality conditions and analogues of
    the diagrams in \cref{def:str}.
\end{lemma}

In general, arrows give rise to \emph{premonoidal categories}~\cite{bobArrows}.
The reason we do not get a monoidal category is that there are two ways to define a map $A(X,Y) \times A(X',Y') \to A(X \otimes X', Y \otimes Y')$, and they need not coincide.
Since we are interested in (symmetric) monoidal categories, we need to restrict to \emph{commutative} arrows\footnote{There is a similar situation with \emph{strong monads}, whose Kleisli category is in general only premonoidal. To get a monoidal category, we need the stronger notion of \emph{commutative monad} (a.k.a.\ \emph{monoidal monad}). See \cref{ex:monad}.}.

\begin{definition}
	\label{def:commutative-arrow}
	Let $\cC$ be symmetric monoidal. An arrow $(A, \pure, \comp, \st)$
	is \emph{commutative} if the following diagram mixing left strength, right strength, and composition commutes:
	\begin{center}
		\begin{tikzcd}[column sep = 60pt]
			A(X,Y) \times A(X',Y') \arrow[dd,"\cong"] \arrow[r,"\st_{X'} \times \st'_Y"]
			& A(X \otimes X', Y \otimes X') \times A(Y \otimes X', Y \otimes Y') \arrow[d,"\comp"]
			\\
			& A(X \otimes X', Y \otimes Y')
			\\
			A(X',Y') \times A(X,Y) \arrow[r,"\st'_X \times \st_{Y'}"]
			& A(X \otimes X', X \otimes Y') \times A(X \otimes Y', Y \otimes Y') \arrow[u,"\comp",swap]
		\end{tikzcd}
	\end{center}
\end{definition}

\noindent
The following is a mild extension of results in the literature \cite{HeunenJ06,RivasJ17}.

\begin{proposition}
\label{thm:arrow-moncat}
Let $\cC$ be symmetric monoidal.
	Every commutative arrow~$A : \Prof(\cC)$
        gives rise to a symmetric monoidal category whose objects are
        the objects of~$\cC$, and whose morphisms from~$X$ to~$Y$ are
        the elements of~$A(X,Y)$. Moreover, there is a strict monoidal functor
        $J : \cC \to A$ where we abuse notation and write $A$ for both
        the arrow and the monoidal category it induces.
\end{proposition}

\begin{example}
	For any symmetric monoidal category~$\cC$, the functor
        $\Hom{\cC} : \cC^\op \times \cC \to \Set$ is an arrow, whose
        associated monoidal category is~$\cC$ itself.
\end{example}

\begin{example} \label{ex:monad}
	Let $T : \cC \to \cC$ be a strong monad on~$\cC$. The
        \emph{Kleisli arrow} of~$T$ is defined by $A_T(X,Y) :=
        \Hom{\cC}(X, TY)$. If $T$ is commutative, then the arrow~$A_T$ is also
        commutative. Its associated monoidal category is the
        Kleisli category~$\cC_T$. Dually, if $T : \cC \to \cC$ is a strong comonad on~$\cC$, its \emph{co-Kleisli arrow} is defined by $A_T(X,Y) := \Hom{\cC}(TX, Y)$.
\end{example}

We use this last example to show lenses form an arrow in the
next section. This example also demonstrates the interest in
arrows from the functional programming community, as arrows provide a single
framework within which computations with both effects and coeffects
can live. To summarise: arrows provide a presentation of monoidal categories
where the monoidal structure is induced by strength, and generalise
both monadic and comonadic computation. We can use constructions
on functors and internal monoids to build new arrows from old; this
flexibility underpins their use in this paper.


\section{Open Games}
\label{sec:basic-games}

Rather than hit the reader with the full profunctor framework we have
developed, we show first how a compositional treatment of open games is
possible, which abstracts away from any specific notion of lens as
the primitive building block for open games. As this is expository
material we only sketch proofs, and work with $\Set$-based open games.

\subsection{Open games, concretely}

Open games were originally presented as follows~\cite{julesPhD}:

\begin{definition}[Open Game]
Let $X, Y , R$ and $S$ be sets. A pure open game $G = (J_G, P_G,
C_G, E_G) : (X, S) \rightarrow (Y, R)$ consists of:
\begin{itemize}
\item A set $J_G$, called the set of strategies of G,
\item A function $P_G : J_G \times X \rightarrow Y$ , called the play function of G,
\item A function $C_G : J_G  \times X \times R \rightarrow S$, called the coutility function of G, and
\item A predicate $E_G : X \times (Y \rightarrow R) \rightarrow \Pow
  (J_G)$, called the equilibrium function of G\footnote{Hedges~\cite{julesPhD} presents open games with best response functions rather than equilibrium functions, but we prefer the presentation using equilibria for simplicity. We will show how to also define open games with best response functions compositionally in \cref{sec:othergames}.} (where $\Pow$ is
  the powerset functor).
\end{itemize}
\end{definition}
An important conceptual step was the 
lens-ification of open games~\cite{GhaniHWZ18}.
The functor $W: \Set \times \Set^{op} \rightarrow
\Set \times \Set^{op}$ defined by $W(X,S) = (X, X \rightarrow S)$ is a
comonad. The category of lenses $\Lens$ is the co-Kleisli category of
$W$: concretely, its objects are pairs of sets, and the morphisms from $(X,S)$ to $(Y,R)$ are pairs of functions $(f,g)$ with $f : X \to Y$ and $g : X \times R \to S$.
 Noting that $X \cong \Lens(1,1) (X,S)$ and $(Y \rightarrow R)
\cong \Lens (Y,R) (1,1)$, we have:

\begin{theorem} Let $I = (1,1)$ and $A$ and $B$ be objects of $\Set \times \Set^{op}$.
An open game $G: A \rightarrow B $ is given by:
(i) a set $J_G$ of strategies;
(ii) a function $p_G:J_G \rightarrow \Lens(A,B)$; and
(iii) a predicate $E_G: J_G \times \Lens(I,A) \times \Lens(B,I) \rightarrow \Bool$.
\end{theorem}

Apart from internalising the notion of open games into structure
internal to the category of lenses, and combining the treatment of
play and coplay into a single morphism, this more abstract formulation
gives cleaner proofs of the following fundamental theorem:

\begin{theorem}
There is a symmetric monoidal category $\Op$ whose objects are pairs of sets, and whose
morphisms are open games, with games with isomorphic strategy sets identified\footnote{An alternative to this quotient would be to work with bicategories instead.}.
This structure reflects the sequential and parallel composition of open games.
\end{theorem}

\subsection{A Compositional Treatment of Open Games}

Rather than proving that open games and similar
structures form a monoidal category in one fell swoop, i.e.\ non-compositionally, we adopt a more compositional, step by step approach.
As mentioned in the introduction we do this by using
profunctors and arrows.
The first step is easy: by \cref{ex:monad}, the comonad $W(X,S) = (X, X \rightarrow S)$
gives rise to an arrow on $\Set \times \Set^\op$, that we also denote by $\Lens$, like its
induced symmetric monoidal category.
Note that the strength of $W$ is defined with respect to the monoidal structure
of $\Set \times \Set^\op$ (not its cartesian structure), which
sends $(X,Y)$ and $(X',Y')$ to $(X\times X', Y
\times Y')$. The unit of this monoidal structure is $I = (1,1)$. 

\begin{proposition}
\label{thm:lens-arrow}
The profunctor $\Lens : \Prof(\Set \times \Set^\op)$ is a commutative arrow.
\end{proposition}

This establishes the well-known fact that the category of lenses is symmetric monoidal.
Further, the unit of this arrow is the core of the lifting of
functions through open games.
The next part of our compositional construction requires that an arrow that satisfies a property called projection at points ---
a property that the key case of $\Lens$ has.
Condition~\eqref{eq:atpoints} below essentially expresses that the projection map  $p^0$ is an 
internal natural transformation. In particular, it implies naturality of $p^0$
in both parameters.

\begin{definition}[Projection at points]
\label{def:proj-at-points}
Let $A: \Prof(\Set \times
\Set^\op)$ be an arrow. Then $A$ has \emph{projection at points} if
there is a family of maps $p^0_{X,X'}: A(I, X \otimes X') \rightarrow A(I,X)$
such that 
\begin{equation}\label{eq:atpoints}  
p^0_{Y,X'}(j  \comp \st_{X'}(a)) = p^0_{X,X'}(j) \comp a  \qquad \mbox{ and } \qquad
p^0_{X,Y'}(j \comp \st'_{X} (a')) = p^0_{X,X'}(j) 
\end{equation}
for all $j \in A(I, X \otimes X')$, $a \in A(X,Y)$ and $a' \in A(X',Y')$.
Using symmetry, this means that there also is a projection 
$p^1_{X,X'}: A(I, X \otimes X') \rightarrow A(I,X')$ given by 
$p^1 _{X,X'} \coloneqq  p^0 _{X,X'} \circ A(\id,\sigma)$ 
that satisfies similar conditions.  
\end{definition}

Given any arrow~$A$ (such as $\Lens$) with projection at points, which we think of as representing some primitive building block,
we now want to attach to it an equilibrium predicate --- the key feature of open games.
The type of this equilibrium predicate is $\Eq(A)(X,Y) := A(I,X) \times A(Y,I) \to \Bool$.
Note that $\Eq(A)$ itself is not an arrow: it has a structure that we call an \emph{$A$-bimodule} in \cref{sec:bimodules}.
However, by bundling together $A$ and $\Eq(A)$, we do get an arrow:

\begin{proposition}
\label{prop:Eq-is-arrow}
Let $A:\Prof(\Set \times \Set^\op)$ be a commutative arrow with
projection at points. Then the profunctor $\WithEq{A} : \Prof(\Set \times \Set^\op)$ defined by
\[
  (\WithEq{A})(X,Y) = A(X,Y) \times \left(A(I,X) \times A(Y,I) \to \Bool\right)
\]
is a commutative arrow.
\end{proposition}
\begin{proof}
The unit of $\WithEq{A}$ is built from the unit
of $A$ and an equilibrium predicate which always returns
true. The composition
\[
  \comp_{X,Y,Z}:(\WithEq{A})(X,Y) \times (\WithEq{A})(Y,Z) \to (\WithEq{A})(X,Z)
\]
is obtained as follows. We are given four arguments 
$s:A(X,Y)$,
$t:A(Y,Z)$, $f:A(I,X) \times A(Y,I) \to \Bool$, and $g:A(I,Y) \times A(Z,I) \to \Bool$, and we define
$$ (s,f) \comp(t,g) \coloneqq \left(s\comp^At, \lambda x \lambda z.\, f(x, t\comp^Az) \wedge g(x\comp^As,z)\right)$$
where $\comp^A$ is the multiplication of $A$. 
%

To define the strength
$\st_{X,Y,Z}: (\WithEq{A})(X,Y) \to (\WithEq{A})(X \otimes Z, Y \otimes Z)$, we are given an
element  $s:A(X,Y)$ and a predicate
$f:A(I,X) \times A(Y,I) \to \Bool$.  We put
\[ \st_Z(s,f) \coloneqq \left(\st_Z^A(s), \lambda x \lambda y.\, f \left(p^0(x), \mathrm{st'}_Y^A(p^1(x))\comp^A y\right)\right) \]
where $\st^A$ denotes the strength of $A$, $\mathrm{st'}^A$ left strength and 
$p^0$ and $p^1$ are the projections at points.
We omit the proof of commutativity of
the arrow.
\end{proof}

Those familiar with the intricacies of the monoidal and sequential  
composition of open games will notice their essence above. What  
is really surprising is that this can be carried out for any set-based  
arrow satisfying projection at points and not just for lenses. This generality was not expected.  
Finally, we turn to indexing by sets of strategies. Again, we work solely with
arrows and introduce an arrow constructor to deal with strategies. The
simplest, used here, is a variant of the fundamental families construction
from categorical logic~\cite[Def.~1.2.1]{jacobs} where we refer to a function $J \rightarrow X$
as a $J$-indexed collection of $X$'s.

\begin{proposition}
  \label{thm:fam-arrow}
  Let $A:\Prof(\Set \times \Set^\op)$ be a commutative arrow. Then the
  profunctor $\Fam(A) : {\Prof(\Set \times \Set^\op)}$ defined by
  \[
    \Fam(A)(X,Y) = \msf{colimit}_{J:\Set_{\text{iso}}} (J \rightarrow A(X,Y))
  \]
    is a commutative arrow.
\end{proposition}

\begin{proof}
The unit of $\Fam(A)$ is the singleton family built from the
unit of~$A$. The composition of $\Fam(A)$ sends $(J, F : J \to A(X,Y))$ and $(K, G : K \to A(Y,Z))$ to $(J \times K, (j, k) \mapsto F(j) \comp^A G(k))$
(see also Hancock's tensor on containers~\cite[Def.~2.3]{hancocksTensor}).
Finally, strength for $\Fam(A)$ takes a family and post-composes
it with the strength of~$A$.
\end{proof}

Note that in the definition of $\Fam(A)$, taking the colimit over the category $\Set_{\text{iso}}$ of sets and bijective functions is crucial.
Indeed, it ensures that families indexed by isomorphic strategy sets are identified.
This quotient is important to prove the associativity of arrow composition; the same issue comes up in the usual treatment of open games~\cite{GhaniHWZ18}.
Combining Propositions \ref{thm:lens-arrow}, \ref{prop:Eq-is-arrow}, \ref{thm:fam-arrow} and \ref{thm:arrow-moncat}, we get a symmetric monoidal category $\Op = \Fam(\WithEq{\Lens})$
and an embedding $\Set \times \Set^\op \rightarrow \Op$, in a compositional manner.


\section{Bimodules and Graded Arrows}
\label{sec:gradings}

We now upgrade the expository illustration of the previous
section. The simplest observation is that there is nothing special
about the category $\Set$ --- any monoidal category can be used. We then
generalise the $\WithEqBare$ construction using a general notion of bimodule
for an arrow in \cref{sec:bimodules}. We decompose $\Fam(A)$
using a notion of \emph{graded} arrow in
\cref{sec:graded-arrows}. In \cref{sec:optics,sec:othergames} we will
use the work in this section to reconstruct open games on optics~\cite{BoltHZ19} and other known and new variations of open games, respectively.

\subsection{Bimodules}
\label{sec:bimodules}

The cumbersome proof of \cref{prop:Eq-is-arrow} can be further decomposed. Notice that the arrow
\[
  (\WithEq{A})(X,Y) = A(X,Y) \times \left(A(I,X) \times A(Y,I) \to \Bool\right) \]
is obtained by combining the arrow~$A$ with the predicate type $\Eq(A)(X,Y) := A(I,X) \times A(Y,I) \to \Bool$. As it turns out, $\Eq(A)$ has a nice structure that we call an $A$-bimodule\footnote{Note that our notion of an $A$-bimodule is equivalent to a profunctor for the (pre)monoidal category induced by $A$ (\cref{thm:arrow-moncat}), matching the alternative name ``bimodule'' for profunctors.}.
We will use this technology to prove \cref{prop:A-with-Eq-is-arrow}, which generalizes \cref{prop:Eq-is-arrow}.

We also take
the opportunity to notice that there is nothing special about the
Booleans in the definition of concrete open games, and replace them with
an arbitrary monoid with suggestively named operations
$(\mC, \mmult, \munit)$ --- Hedges' software
implementation~\cite{julesSoftware} uses this insight to collect more
in-depth diagnostics when checking equilibria.

\begin{definition}[Bimodule]
\label{def:bimodule}
Let $A : \cC \profto \cC$ be an arrow. An \emph{$A$-bimodule} is a profunctor $B : \cC \profto \cC$ equipped with natural transformations $\ell : A \bullet B \to B$ and $r : B \bullet A \to B$, such that:
\begin{enumerate}
\item $\ell$ is a left-action, i.e., the two diagrams below commute;
\begin{mathpar}
	\begin{tikzcd}
	A \bullet (A \bullet B) \arrow[rr,"\alpha"] \arrow[d,"A \bullet \ell",swap]
	&& (A \bullet A) \bullet B \arrow[d,"\comp \bullet B"] \\
	A \bullet B \arrow[r,"\ell",swap] & B & A \bullet B \arrow[l,"\ell"]
	\end{tikzcd}
	\and
	\begin{tikzcd}
	\Hom{\cC} \bullet B \arrow[r,"\lambda"] \arrow[d,"\pure \bullet B",swap] & B\\
	A \bullet B \arrow[ur,"\ell",swap]
    \end{tikzcd}
\end{mathpar}
\item $r$ is a right-action (similar diagrams);
\item $\ell$ commutes with $r$, i.e.,
\begin{mathpar}
	\begin{tikzcd}
		A \bullet (B \bullet A) \arrow[rr,"\alpha"] \arrow[d,"A \bullet r",swap]
		&& (A \bullet B) \bullet A \arrow[d,"\ell \bullet A"] \\
		A \bullet B \arrow[r,"\ell",swap] & B & B \bullet A \arrow[l,"r"]
	\end{tikzcd}
\end{mathpar}
\end{enumerate}
When $B$ is a strong profunctor, $\ell$ and~$r$ are also required to be strong natural transformations.
\end{definition}


\begin{definition}
	Let $A$ be an arrow and $B : \StrProf(\cC)$ an $A$-bimodule.
	We say that~$B$ is \emph{commutative} when the following diagram (akin to the one of Definition~\ref{def:commutative-arrow}) commutes.
	\begin{center}
		\begin{tikzcd}[column sep = 60pt]
			A(X,Y) \times B(X',Y') \arrow[dd,"\cong"] \arrow[r,"\st_{X'} \times \st'_Y"]
			& A(X \otimes X', Y \otimes X') \times B(Y \otimes X', Y \otimes Y') \arrow[d,"\ell"]
			\\
			& B(X \otimes X', Y \otimes Y')
			\\
			B(X',Y') \times A(X,Y) \arrow[r,"\st'_X \times \st_{Y'}"]
			& B(X \otimes X', X \otimes Y') \times A(X \otimes Y', Y \otimes Y') \arrow[u,"r",swap]
		\end{tikzcd}
	\end{center}
\end{definition}

Recall that profunctors inherit the cartesian structure of $\Set$.
Given $F,G : \cC \profto \cD$, their product is defined by $(F \times G)(c,d) := F(c,d) \times G(c,d)$; and the terminal profunctor $1 : \cC \profto \cD$ is the constant functor that picks out a singleton set.
It is easy to check that this is also a cartesian structure on~$\StrProf(\cC)$.
Finally, note that $\bullet$ distributes over $\times$, in the sense that there is a strong natural transformation of type $F \bullet (G \times H) \to (F \bullet G) \times (F \bullet H)$, and similarly with $F$ on the right. \fullappendix{See \cref{app:StrProf-cartesian} for details.}

\begin{definition}[Bimodule with monoid]
\label{def:bimodule-monoid}
When $B$ additionally has a monoid structure $(B,m,e)$ w.r.t.\ the cartesian product $(\times, 1)$ in $\StrProf(\cC)$, we also require that~$\ell$ and~$r$ preserve the monoid structure of~$B$:
		\begin{mathpar}
			\begin{tikzcd}
				A \bullet (B \times B) \arrow[rr,"\text{distributivity}"] \arrow[d,"A \bullet m",swap]
				&& (A \bullet B) \times (A \bullet B) \arrow[d,"\ell \bullet \ell"] \\
				A \bullet B \arrow[r,"\ell",swap] & B & B \times B \arrow[l,"m"]
			\end{tikzcd}
			\and
			\begin{tikzcd}
				A \bullet 1 \arrow[r,"\id \bullet e"] \arrow[d,"!",swap] & A \bullet B \arrow[d,"\ell"]\\
				1 \arrow[r,"e",swap] & B
			\end{tikzcd}
		\end{mathpar}
\end{definition}

\begin{theorem}
\label{thm:bimodule-arrow}
\label{thm:bimodule-arrow-comm}
Assume $A : \cC \profto \cC$ is an arrow, and $B : \StrProf(\cC)$ is a monoid w.r.t.\ cartesian product $(\times, 1)$ and an $A$-bimodule, then $A \times B$ is an arrow.
Moreover, if both~$A$ and~$B$ are commutative (as an arrow, as a monoid and as a bimodule), then so is $A \times B$.
\end{theorem}
\begin{proof}
Since~$A$ and~$B$ are strong profunctors, then so is $A \times B$.
Moreover, our hypothesis gives us the following strong natural transformations:
\begin{align*}
\pure^A &: \Hom{\cC} \to A & \ell &: A \bullet B \to B & e &: 1 \to B \\
\comp^A &: A \bullet A \to A & r &: B \bullet A \to B & m &: B \times B \to B
\end{align*}
Then we define $\pure^{A \times B} := \pure^A \times e$, and $\comp^{A \times B}$ is given by
\begin{center}
\begin{tikzcd}
	(A \times B) \bullet (A \times B)
	  \arrow[rr,"\text{distributivity}"] & &
          \parbox{3.2cm}{\centering
            $(A \bullet A) \times (A \bullet B)$ \\
            $\times (B \bullet A) \times (B \bullet B)$}
	  \arrow[rr,"\comp^A\, \times\, \ell\, \times\, r\, \times\, !"] &&
	A \times B \times B
	  \arrow[r,"\id\, \times\, m"] &
	A \times B
\end{tikzcd}
\end{center}
One can then check that the arrow laws are satisfied\fullappendix{ (see Appendix~\ref{app:bimodules})}.
\end{proof}

%

In virtue of \cref{thm:bimodule-arrow-comm}, the proof of \cref{prop:Eq-is-arrow} now amounts to showing that $\Eq(A)$ is a commutative $A$-bimodule with a $(\times, 1)$-monoid structure. In fact, we can further decompose this proof by introducing the following notion of \emph{context}. A similar notion of context appears in Hedges~\cite{JulesGameSemantics}, where it is axiomatized at the level of the monoidal category induced by the arrow~$A$.

\begin{definition}[Context]
\label{def:context}
Let $A : \cC \profto \cC$ be an arrow. A profunctor $B : \cC \profto \cC$ is called a  \emph{context} when it is a commutative $A$-bimodule equipped with a \emph{costrength} $\cst_{XYZ} : B(X \otimes Z, Y \otimes Z) \to B(X,Y)$, satisfying some coherence diagrams similar to the ones for strength (full details in \cref{app:costrength}).
\end{definition}

The notion of costrength is a weakening of the notion of a trace for a symmetric monoidal category to arbitrary profunctors --- in particular if the profunctor $B$ is the $\Hom{}$-functor of a traced monoidal category, then $B$ has a canonical costrength.
The meat of \cref{prop:Eq-is-arrow} lies in the next \lcnamecref{thm:ctxt-bimodule}. Notice that the position of~$X$ and~$Y$ is reversed; hence the pre-composition with $\flip : \cC^\op \times \cC \to \cC \times \cC^\op$ in \cref{thm:M-to-the-B-flip}.
Another example of context, which does not require projection at points, will be given later in \cref{def:optic-context}.

\begin{lemma}
	\label{thm:ctxt-bimodule}
	If $A$ has projection at points, then $\Ctx_A(X,Y) = A(I,Y) \times A(X,I)$ is a context.
\end{lemma}
\begin{proof}
  The left action is given by the family $\ell_{XYZ} : A(X,Y) \times \Ctx_A(Y,Z) \to \Ctx_A(X,Z)$, defined by:
  \[
    \ell_{XYZ} := (s,z,y) \mapsto (z, s \comp^A y) : A(X,Y) \times A(I,Z) \times A(Y,I) \longrightarrow A(I,Z) \times A(X,I)
  \]
  The right action is similar. To define the costrength, we use the strength of $A$ and the projections $p^0, p^1$:
  \[
  \cst_{XYZ} := (y,x) \mapsto (p^0(y), \st'^A_X(p^1(y)) \comp^A x) : A(I, Y \otimes Z) \times A(X \otimes Z, I) \longrightarrow A(I,Y) \times A(X,I)
  \qedhere
  \]
\end{proof}

\begin{theorem}
\label{thm:M-to-the-B-flip}
Let $A$ be a commutative arrow, $B : \cC \profto \cC$ be a context for~$A$, and
  $(M,\mmult,\munit)$ be a commutative monoid (in $\Set$). Then $M^{B\, \circ\, \flip}$
  is a commutative $A$-bimodule and $(\times,1)$-monoid.
\end{theorem}
\begin{proof}
	The left and right actions are given by the families of functions
	\begin{align*}
	\ell_{XYZ} :\; & A(X,Y) \times M^{B(Z,Y)} \longrightarrow M^{B(Z,X)} &
	r_{XYZ} :\; & M^{B(Y,X)} \times A(Y,Z) \longrightarrow M^{B(Z,X)} \\
	&(a,h) \longmapsto \lambda b.\; h(r^B\,b\,a) &
	&(h,a) \longmapsto \lambda b.\; h(\ell^B\,a\,b)
	\end{align*}
	The $(\times, 1)$-monoid structure is inherited from the one of~$M$:
	\begin{align*}
	e_{XY} :\; & 1 \longrightarrow M^{B(Y,X)} &
	m_{XY} :\; & M^{B(Y,X)} \times M^{B(Y,X)} \longrightarrow M^{B(Y,X)} \\
	&\ast \longmapsto \lambda b.\; \munit &
	&(h,h') \longmapsto \lambda b.\; h\,b \mmult h'\,b
	\end{align*}
	Finally, the strength $\st : M^{B(Y,X)} \longrightarrow M^{B(Y \otimes Z,X \otimes Z)}$ is obtained by pre-composing with~$\cst^B$.
	All the bimodule axioms of $M^{B\, \circ\, \flip}$ can then be easily derived from those of~$B$.
\end{proof}

\subsection{Graded Arrows and graded bimodules}
\label{sec:graded-arrows}

The $\Fam(A)$ construction is really a two step construction: it takes
an arrow $A$ and parameterises by elements of sets $J$, then the
parameter sets are hidden using a colimit. The intermediate step in
this process is not an arrow because its composition composes
morphisms with different parameter sets. However, it is a
\emph{graded} arrow, which we define by adapting the definition of
graded monads~\cite{katsumata2014parametric}. Graded monads have been used to model programming
languages with variable sets of possible side effects. The
multiplication of a graded monad is a natural transformation
$T_pT_q \Rightarrow T_{pq}$, where the $p,q$ are from a symmetric
monoidal category $(\cP, i, \cdot)$ of grades. Graded monads are lax
monoidal functors $\cP \to \mathrm{Endo}(\cC)$. Swapping endofunctors
for strong profunctors yields:
\begin{definition}
  A $\cP$-\emph{graded arrow} is a lax monoidal functor
  $A : \cP \to \StrProf(\cC)$.
\end{definition}

Unfolding this definition shows that a graded arrow comes equipped
with graded composition $A_p(X,Y) \times A_q(Y,Z) \to A_{pq}(X,Z)$ and
unit $\Hom{\cC}(X,Y) \to A_i(X,Y)$ that use the monoidal structure of
$\cP$. Strength $A_p(X,Y) \to A_p(X \otimes Z, Y \otimes Z)$ is
provided for every grade $p$, but when we use strength to define
parallel composition we will use the monoidal structure of grades. To
make sure that parallel composition makes sense, graded arrows must be
commutative. Note that the symmetry of $\cP$'s monoidal structure is
required for this definition to make sense.

\begin{definition}[Commutative Graded Arrow]
  A $\cP$-graded arrow $A$ is \emph{commutative} if the two functions
  $A_p(X,Y) \times A_q(X',Y') \to A_{pq}(X \otimes X', Y \otimes Y')$
  defined using left and right strength are equal.
\end{definition}

When a graded arrow is commutative we obtain a canonical parallel
composition of morphisms that commutes with sequential
composition. Note that the composition of parameters is the same as in the
sequential case, forced by an Eckmann-Hilton style argument.


We now use graded arrows to reconstuct $\Fam(A)$. The first step is
this proposition, which follows because the graded arrow is
constructed pointwise.
\begin{proposition}
  \label{prop:param-arrow}
  Let $(\cP, i, \cdot)$ be a symmetric monoidal subcategory of
  $(\Set, 1, \times)$, and let $A : \StrProf(\cC)$ be an arrow. Then
  $(- \to A)_p(X,Y) = p \to A(X,Y)$ is a $\cP^\op$-graded arrow. If $A$ is
  commutative, then so is $- \to A$.
\end{proposition}

The second step in the $\Fam(A)$ construction is the following result,
which shows that there is a general way of ``summing out'' the grading
of a graded arrow to recover an (ungraded) arrow.
\begin{proposition}
  \label{thm:colim-graded-arrow}
  For every graded arrow $A : \cP \to \StrProf(\cC)$,
  $(\Hide{A})(X,Y) = \colim_p~A_p(X,Y)$ is an arrow. If $A$ is
  commutative, then so is $\Hide{A}$.
\end{proposition}

\begin{proof}
  $\Hide{A}$ is certainly a profunctor, so it remains to construct the
  monoid operations. The unit is defined by the composite
  $\Hom{\cC}(X,Y) \stackrel{\pure_i}\to A_i(X,Y) \stackrel\iota\to \colim_p
  A_p(X,Y)$ , where $\iota$ is the injection into the
  colimit. Composition is defined using colimits commuting with
  products in $\Set$:
  $(\colim_p F_p(X,Y)) \times (\colim_q F_q(Y,Z)) \cong \colim_{p,q}
  \left(F_p(X,Y) \times F_q(Y,Z)\right) \stackrel{\colim
    \comp_{p, q}}\longrightarrow \colim_{p,q} F_{pq}(X,Z) \to \colim_r
  F_r(X,Z)$. Strength is defined as
  $\colim_p F_p(X,Y) \stackrel{\colim \st}\longrightarrow \colim_p
  F_p(X \otimes Z, Y \otimes Z)$. Each of the axioms follows from the
  corresponding axiom for the graded arrow $A$.
\end{proof}

Thus we have decomposed the $\Fam(A)$ construction into
$\Hide{(- \to A)}$, using the intermediate concept of graded
arrow. We will also use this technology in \cref{thm:optics-arrow} to prove that optics form an arrow, as well as in \cref{sec:othergames} to introduce other indexing operators besides~$\Fam$. For variations such as best-response games and probabilistic games, we need to introduce a graded version of bimodules:

\begin{definition}[Graded bimodule]
  \label{defn:graded-bimodule}
  Let $A : \cP \to \StrProf(\cC)$ be a graded arrow. A \emph{graded
  $A$-bimodule} is a functor $B : \cP \to \Prof(\cC)$ equipped with natural transformations $\ell_{pq} : A_p \bullet B_q \to B_{pq}$ and
  $r_{pq} : B_p \bullet A_q \to B_{pq}$ natural in $p$
  and $q$, satisfying the bimodule axioms:
  \begin{enumerate}
  \item $\ell$ is a left action,
  \item $r$ is a right action,
  \item $\ell$ commutes with $r$.
  \end{enumerate}
When $B : \cP \to \StrProf(\cC)$,  we moreover require $\ell_{pq}$ and $r_{pq}$ to be strong natural transformations.
As in \cref{sec:bimodules}, we can define the notion of a \emph{commutative} graded bimodule; and when each~$B_p$ has a $(\times, 1)$-monoid structure, we require that~$\ell$ and~$r$ should preserve it. \fullappendix{See \cref{app:graded-bimodules} for full details.}
\end{definition}

\noindent
The following theorem is the graded generalisation of
\cref{thm:bimodule-arrow}.
\begin{theorem}
  \label{thm:graded-bimodule-arrow}
  Assume $A : \cP \to \StrProf(\cC)$ is a commutative graded arrow, and
  $B : \cP \to \StrProf(\cC)$ is a commutative graded $A$-bimodule, such that $B_p$ is a commutative $(\times, 1)$-monoid for every $p$. Then $A \times B : \cP \to \StrProf(\cC)$ is a commutative graded arrow.
\end{theorem}
\begin{proof}
  The proof is very similar to the proof of Theorem~\ref{thm:bimodule-arrow}. 
\end{proof}


\section{Open games based on optics}
\label{sec:optics}

In this section, we use graded arrows and bimodules compositionally define open games starting
from a monoidal category, so as to treat open games containing
computational effects via the Kleisli category. We fundamentally
achieve the same result as \cite{BoltHZ19}, but we do so in a
compositonal manner. Lenses only define a monoidal category when the
underlying category is cartesian. In a monoidal category we work with
optics~\cite{pastroStreet2008,milewski,catOptics,gibbonsPearl}
instead:
\begin{proposition}
  \label{thm:optics-arrow}
  Let $\cC$ be symmetric monoidal, and $A : \Prof(\cC)$ a commutative
  arrow. Then the profunctor $\Optic{A} : \Prof(\cC \times \cC^\op)$
  defined by
  \[
    \Optic{A}((X, S), (Y, R)) = \int^{P\in\cC} A(X, P \otimes Y) \times A(P \otimes R, S)
  \]
  is also a commutative arrow.
\end{proposition}
\begin{proof}
  By \cref{thm:colim-graded-arrow}, and using that coends can
  be reduced to colimits using the twisted arrow category
  $\msf{Tw}(\cC^\op)^\op$ (with objects arrows $f : P' \to P$ in $\cC$)~\cite[Remark 1.2.3]{coendcalc}, it is enough
  to show that the assignment
  $\Optic{A}^{\msf{Tw}} : \msf{Tw}(\cC^\op)^\op \to \StrProf(\cC \times
  \cC^\op)$ given by
  \[
\Optic{A}^{\msf{Tw}}\;(P \xleftarrow{f} P', (X, S), (Y, R)) = A(X, P' \otimes Y) \times A(P \otimes R, S)
  \]
  is a commutative graded arrow, which follows from $A$ being a
  commutative arrow and $\cC$ being symmetric monoidal.
\end{proof}

This compositional proof has factored out the
coend calculations involved into a separate statement about graded
arrows.  We expect many other monoidal categories defined in similar
ways, e.g.\ combs~\cite{combs}, can be proved monoidal using graded arrows to
similarly avoid coend manipulation.
Next, we follow Bolt, Hedges and Zahn~\cite[\S 3.4]{BoltHZ19} in generalising the type of
the equilibrium function as needed when working in a monoidal setting.

\begin{definition}
\label{def:optic-context}
  Let $A : \Prof(\cC)$ be an arrow, and $\mC$ a monoid. We define the \emph{context} and \emph{equilibrium type} profunctors $\Ctxt{A}, \EqC{A} : \Prof(\cC)$ by
  \begin{align*}
    \Ctxt{A}(X, Y) &= \Optic{A}((I, I),(Y,X)) \\
    \EqC{A}(X, Y) &=  \Ctxt{A}(Y, X) \to \mC
  \end{align*}
\end{definition}

We now prove the optics analogue of \cref{thm:ctxt-bimodule}:

\begin{lemma}
  \label{thm:optics-ctxt-bimodule}
  Let $A : \Prof(\cC)$ be a commutative arrow. Then $\Ctxt{A}$ is a
  context, i.e.\ a commutative $A$-bimodule coherently equipped with a costrength.
\end{lemma}
\begin{proof}
  The left costrength is given by the $\Theta$-indexed family of maps $(\cst'_Z)_\Theta = \iota_{\Theta \otimes Z}$.
  The actions are defined using composition and the strength of~$A$, e.g.\ for the left action $A(X, Y) \times \Ctxt{A}(Y, Z) \to \Ctxt{A}(X, Z)$:
  \begin{center}
  \begin{tikzcd}[column sep = 2cm,ampersand replacement=\&]
    {\begin{array}[t]{@{}c@{\hspace{0.3em}}l}
       & A(X, Y) \times \Ctxt{A}(Y, Z)\\
       \cong & \int^{P} A(I, P \otimes Z) \times A(X, Y) \times A(P \otimes Y, I)
     \end{array}}
    \arrow{d}[swap]{\int^P(\id \times \st'_P \times \id)} \\
    \int^{P} A(I, P \otimes Z) \times A(P \otimes X, P \otimes Y) \times A(P \otimes Y, I)
    \arrow{r}{\int^P(\id \times \comp)} \&
    \Ctxt{A}(X,Z)
  \end{tikzcd}
\end{center}
One can then check that all the required laws of a bimodule with costrength are satisfied\fullappendix{ (see \cref{sec:proofs-optics})}.
\end{proof}

Applying \cref{thm:bimodule-arrow} together with
\cref{thm:M-to-the-B-flip}, we thus immediately get the optics
analogue of \cref{prop:Eq-is-arrow}, but this time we do not need the
assumption of projection at points:

\begin{proposition}
\label{prop:A-with-Eq-is-arrow}
  Let $A : \Prof(\cC)$ be a commutative arrow and $\mC$ a monoid.
  The profunctor $\WithEqC{A} : \Prof(\cC)$ defined by
  \[
    (\WithEqC{A})(X,Y) := A(X,Y) \times \EqC{A}(X, Y)
  \]
  is a commutative arrow.
\end{proposition}

To parameterise by strategies, we use graded arrows from
Section~\ref{sec:graded-arrows}.  By \cref{prop:param-arrow} and
\cref{thm:colim-graded-arrow}, when $A$ is an arrow,
$\Fam(A) = \Hide{(- \to A)}$ is an arrow. Composing the steps, we get
the following theorem. By using the abstract framework in the previous
section, the proofs are much shorter.

\begin{theorem}
  Let $\cC$ be a symmetric monoidal category and $\mC$ a monoid. There
  is a symmetric monoidal category
  $\Op_{\cC} = \Fam(\WithEqC{\Optic{\Hom{\cC}}})$ of open games
  in $\cC$, and an embedding
  $\cC \times \cC^\op \rightarrow \msf{Op}_\cC$.
\end{theorem}


\section{Other variants of open games}
\label{sec:othergames}

We demonstrate the power of our framework by showing how a host of
variations of open games can be easily proven to form monoidal
categories.

\begin{example}[Fibred Open games]
Lets say we want to add a little bit of type dependency so coutility
can depend on state and utility on the moved played. Type dependency
is modelled fibrationally~\cite{jacobs},  e.g.\ via the families fibration, which is split. Given a (split) fibration $p:\cal{E}
\rightarrow \cal{B}$, define the monoidal category of fibrational
lenses $\Lens_p$ to have the same objects as $\cal{E}$, and with morphisms
from $E$ to $E'$ given by pairs $(f, g)$, where $f : pE \rightarrow pE'$ is a morphism in $\cal{B}$ and $g:f^*E' \rightarrow E$ is a morphism in~$\cal{E}$.
If $p$ is a symmetric monoidal fibration~\cite{shulman2008monfib} then $\Lens_p$ is a commutative arrow.
Defining $\Fam$ and $\WithEqCBare$ as in \cref{sec:optics}, we hence get an arrow $\Fam(\WithEqC{\Lens_{p}})$ of fibrational open games. This subsumes the use of
container morphisms~\cite{containers} as dependently typed open games. The point of this example is not that this is necessarily the
right way to do dependent games --- clearly a lot more dependency is
possible --- but rather to show how new variants of open games can easily
be shown monoidal in our compositional framework.
\end{example}




\begin{example}[Best Response]
	\label{ex:bestresp}
  Using graded bimodules (Definition \ref{defn:graded-bimodule}) and
  Theorem \ref{thm:graded-bimodule-arrow}, we can recover \emph{best
    response} games~\cite[\S 2.1.5]{julesPhD}. In these games, the
  equilibrium predicate is a relation on strategies in context,
  expressed as a graded bimodule w.r.t.\ the graded arrow $A_J := J \to A$, for a given arrow $A$.
  \begin{displaymath}
    \mathrm{BestResp}_J(X,Y) = \Ctxt{A}(Y,X) \to J \times J \to \mathrm{Bool}
  \end{displaymath}
  The left action, of type $A_J(X,Y) \times \mathrm{BestResp}_{J'}(Y,Z) \to \mathrm{BestResp}_{J\times J'}(X,Z)$, can be defined as follows.
  Given a context $c \in \Ctxt{A}(Z,X)$ and a strategy $j \in J$, we can use the right action of the context (which is an $A$-bimodule) to get a context $c' \in \Ctxt{A}(Z,Y)$.
  Thus, given four parameters $(j_1,j'_1,j_2,j_2')$, we use~$j_1$ to create a new context and feed it to $\mathrm{BestResp}_{J'}(Y,Z)$ together with $(j'_1,j'_2)$ to obtain a boolean.
  Note that~$j_2$ is discarded, which is also what happens when we define the sequential composition of best response games.
  Best response games are then defined as
  $\colim_J\big((J \to \Lens(X,Y)) \times \mathrm{BestResp}_J(X,Y)\big)$,
  which automatically gives a symmetric monoidal category by Theorem
  \ref{thm:graded-bimodule-arrow} and Proposition
  \ref{thm:arrow-moncat}.
\end{example}

\begin{example}[Probabilistic games]
  \label{ex:probOG}
  Surprisingly, Probabilistic Open Games~\cite{probOG} can be recovered using a construction similar to the one of \cref{ex:bestresp}.
  The main difference is that the graded bimodule is now indexed by $\D(J)$, where $\D$ is the finite distribution monad.
  However, for it to work properly, we need to make some assumptions on~$\cC$ and~$A$.
  
  For base category, we take $\cC = \Set \times (\DAlg)^\op$, where $\DAlg$ is the category of algebras for the monad~$\D$.
  There is a forgetful functor $U : \Set \times (\DAlg)^\op \to \Set \times
  \Set^\op$ to the ``usual'' base category $\Set \times \Set^\op$, obtained by sending an algebra to its carrier set.
  Over this category, the profunctor $\Lens(U-,U-) = \Lens \circ (U^\op \times U)$ can easily be seen to be a commutative arrow (because $\Lens$ is), and thus by \cref{prop:param-arrow},
  $A_J(X,Y) = J \to \Lens(U(X),U(Y)$ is a $\Set_{\text{iso}}$-graded arrow.
  With these choices of $\cC$ and $A$, the graded $A$-bimodule for probabilistic equilibria selection is:
  \begin{displaymath}
    \mathrm{ProbEquib}_J(X,Y) = \Ctxt{\Lens}(U(Y),U(X)) \to \D(J) \to \mathrm{Bool}
  \end{displaymath}
  To match Definition 12 in~\cite{probOG}, the left and right actions are defined quite differently: while the left action relies on the so-called \emph{Kleisli predicate lifting} $\hashbar$~\cite[Def.~11]{probOG}, the right action uses the $\D$-algebra structure to compute the expectation of the payoff.
  
  Finally, we define
  $\Op_{\mathrm{Prob}} = \colim_J\big(A_J(X,Y) \times \mathrm{ProbEquib}_J(X,Y)\big)$,
  which again automatically yields a symmetric monoidal category by
  our results. This links nicely to the definition of the equilibrium
  for sequential composition of probabilistic open games~\cite[Def.~12]{probOG}: the definition in \emph{loc.cit.}\
  uses a conjunction whose conjuncts correspond to the left  and right bimodule actions, respectively. This matches the arrow composition of $\Op_{\mathrm{Prob}}$ given by \cref{thm:graded-bimodule-arrow} (cf. proof of \cref{thm:bimodule-arrow-comm}). 
\end{example}

So far all examples have had sets of strategies. Its not surprising
this does not have to be the case and there are other options. Our
framework means this can be handled compositionally.

\begin{example}[Strategies need not be Sets]
	Consider the parametrisation functor $\msf{Para} : \Prof(\cC) \to \Prof(\cC)$ defined by
	$(\msf{Para} \; A)(X,Y) = \msf{colim}_{J  : \cC_{\text{iso}}} \; A(J \otimes
	X,Y)$~\cite{backprop,brunoMSc}. By our theorem on colimits
	of graded arrows, $\msf{Para}$ maps commutative arrows to
	commutative arrows and hence could replace
	$\Fam$. More generally, given $F:\cB \times \cC
	\rightarrow \cC$, we define $(\msf{Para}_F\; A) (X,Y) = \msf{colim}_{J
		: \cB_{\text{iso}}} \; A( F(J,X),Y)$.
	Again $\msf{Para}_F$ maps commutative arrows to
	commutative arrows, and can hence be used to construct open games
	$\msf{Para}_F(\WithEqC{\Lens})$ with strategies living in a different category than the play/coplay function.
\end{example}

\begin{example}[Learners]
If we take the category of lenses, omit the equilibria (or, equivalently, take equilibria valued in the trivial monoid), and use the $\msf{Para}$
construction above, we get Fong, Spivak and Tuy\'eras' monoidal category of
learners~\cite{backprop}. That is, defining $\msf{Learn} = \msf{Para}(\Lens)$
makes learners an instance of our framework.
\end{example}

\begin{example}[Effectful coplay and utility functions]
  Let $\cC$ be a symmetric monoidal category. In \cref{sec:optics}, we
  used optics to define an arrow on $\cC \times \cC^\op$,
  but the coend machinery used in its definition is quite heavy, e.g.\ 
one might worry about the existence of the coends. An alternative construction (folklore according to \cite[\S
  2.2]{spivakGenLens}), instead defines an
  arrow on $\msf{Comon}(\cC) \times \cC^\op$, where
  $\msf{Comon}(\cC)$ is the category of \emph{commutative comonoids}
  in $\cC$. The underlying profunctor is defined by
  $\Lens_{\cC}((X, S), (Y, R)) = \Hom{\msf{Comon}(\cC)}(X, Y) \times \Hom{\cC}(X
  \otimes R, S)$, i.e.\ the same as for ordinary lenses, except that
  the play functions must respect the comonoid structure. In
  particular, If $\cC$ is a Kleisli category for a commutative monad
  on a cartesian category, then every object has a natural comonoid
  structure. Further, in this case $\Lens_{\cC}$ has projections at
  points, and so we can use \cref{thm:ctxt-bimodule} and
  \cref{thm:M-to-the-B-flip,thm:bimodule-arrow-comm} to equip it with
  a notion of equilibrium function $\Lens_{\cC} \times E(\Lens_{\cC})$. In the resulting notion of games $\Fam(\Lens_{\cC} \times E(\Lens_{\cC}))$, comonoid
  homomorphisms are intuitively `pure' maps, so this notion of game allows
  `effectful` coplay and utility functions, while the play function
  and the state in the equlibrium predicate remain pure.
\end{example}








\section{Conclusions and future work}

We have developed a new framework for open games.
Conceptually, our framework is not just an axiomatisation
of open games, but a compositional framework allowing us to vary
and reorder components without having to redevelop the whole theory of
the resulting open games. This is achieved using standard
mathematical structures: arrows, bimodules and gradings. These
structures allow us to define new arrows from old which is the hallmark of compositionality. Collectively, we demonstrate the
utility of this framework by showing how  a number of existing and new
variants of open games can be shown to form monoidal categories with
relative ease. This is important as the monoidal structure reflects
the parallel and sequential composition of open games.

As is becoming clear, there is a huge variety of open games each with
their pros and cons and each of which makes sense to use in some
circumstances. We must now work to determine which of the categories
of Open Games we have constructed are useful for Compositional Game
Theory. To do this, we will need to gain experience in modelling game
theoretic concepts in our categories, and to observe how modifying the
constructions using the building blocks we have developed here affects
the game theoretic models. We will be aided in this investigation by
studying embeddings between the various categories of Open Games.

We have built tools to prove that categories are symmetric
monoidal. However, there are other operators on games beyond parallel
and sequential composition, e.g.\ an operator for repeated games is
needed. A start is given by~\cite{GhaniKLF18}, but we want to
integrate such operators into our general compositional framework.

All our profunctors have been of the form
$\cC^\op \times \cC \to \Set$ but we could replace $\Set$ by some
other cocomplete monoidal category with (we conjecture) minimal
disruption. This would give us an enriched theory of open games which
will be important for the investigation of topological games, with
notions of limit games.

Finally, it will be interesting to see if this framework has
ramifications for software for open games by providing all the
benefits compositionality does. That is, rather than implementing a
specific notion of open games and expecting users to use that notion,
our compositional framework would allow users to build their own
notion of open games suitable for their own needs (mixing both
predefined components such as for example $\msf{Fam}$ or $\msf{Para}$
with their own specific definitions). The fact that we have used
canonical categorical structures such as profunctors, modules and
gradings means the road to a Haskell implementation should be a
pleasant road to engage with. We hope to make it so as a matter of
priority.

\paragraph{Acknowledgements.} The authors would like to thank Exequiel Rivas and Jules Hedges for helpful discussions and comments.


\bibliographystyle{eptcs}
\bibliography{papers}

\newpage

\appendix
\section{Definition of a bimodule with costrength}
\label{app:costrength}

For the sake of completeness, we spell in full details the definition of a costrength, and of a bimodule equipped with a costrength (used to define the notion of context, \cref{def:context}).

\begin{definition}[Costrength]
\label{def:costr}
A \emph{(right) costrength} for a profunctor $F : \cC \profto \cC$ is
a family of morphisms
$\cst_{X,Y,Z} : F(X \otimes Z, Y \otimes Z) \to F(X,Y)$
(natural in~$X,Y$ and dinatural in~$Z$) making two diagrams
commute:
\begin{center}
	\begin{tikzcd}
		F((X \otimes Z) \otimes Z', (Y \otimes Z) \otimes Z') \arrow[r,"\cst_{Z'}"] \arrow[d,"{F(\alpha^{-1},\alpha)}"]
		& F(X \otimes Z, Y \otimes Z) \arrow[d,"\cst_{Z}"]
		\\
		F(X \otimes (Z \otimes Z'), Y \otimes (Z \otimes Z')) \arrow[r,"\cst_{Z\otimes Z'}",swap]
		& F(X,Y)
	\end{tikzcd}
    \qquad and \qquad
	\begin{tikzcd}
	F(X \otimes I, Y \otimes I) \arrow[d,"\cst_I",bend left] \arrow[d,"{F(\rho^{-1},\rho)}",bend right,swap] \\
	F(X,Y)
\end{tikzcd}
\end{center}
Given a profunctor $F : \cC \profto \cC$ with a right costrength, assuming $\cC$ is symmetric, we can define the left costrength $\cst'_{X,Y,Z} : F(Z \otimes X, Z \otimes Y) \to F(X,Y)$ by $\cst' := \cst \circ F(\sigma, \sigma)$.
\end{definition} 

When an $A$-bimodule has a strength, the two actions $\ell$ and $r$ are required to be strong natural transformations. When it has a costrength, we require similar conditions:
\begin{definition}[Bimodule with costrength]\label{def:bimod_costr}
Let $A : \cC \profto \cC$ be an arrow.
When an $A$-bimodule $B : \cC \profto \cC$ is equipped with a costrength, we also require the two commutative diagrams below, relating the strength of~$A$, costrength of~$B$ and the left and right actions~$\ell$ and~$r$:
\begin{center}
	\begin{tikzcd}
		A(X \otimes W, Y \otimes W) \times B(Y \otimes W, Z \otimes W) \arrow[rr,"\ell"]
		&& B(X \otimes W, Z \otimes W) \arrow[dd,"\cst_W"] \\
		A(X,Y) \times B(Y \otimes W,Z \otimes W)
		\arrow[u,"\st_W \times \id"] \arrow[d,"\id \times \cst_W",swap] && \\
		A(X, Y) \times B(Y,Z) \arrow[rr,"\ell"]
		&& B(X,Z)
	\end{tikzcd}
\end{center}
and
\begin{center}
\begin{tikzcd}
	B(X \otimes W, Y \otimes W) \times A(Y \otimes W, Z \otimes W) \arrow[rr,"r"]
	&& B(X \otimes W, Z \otimes W) \arrow[dd,"\cst_W"] \\
	B(X \otimes W, Y \otimes W) \times A(Y,Z)
	\arrow[u,"\id \times \st_W"] \arrow[d,"\cst_W \times \id",swap] && \\
	B(X, Y) \times A(Y,Z) \arrow[rr,"r"]
	&& B(X,Z)
\end{tikzcd}
\end{center}
\end{definition}
\medskip

\noindent
Finally, we also have a counterpart of the commutativity condition, mixing strength and costrength.

\begin{definition}[Commutative bimodule with costrength]\label{def:comm_bimod_costr}
	An $A$-bimodule $B$ with a costrength is commutative when the following diagram commutes:
\begin{center}
\begin{tikzcd}[column sep = -5pt]
	A(X,Y) \times B(Y \otimes X', X \otimes Y')
	\arrow[rr,"\cong"]
	\arrow[d,"\st_{X'} \times \id",swap]
	&& B(Y \otimes X', X \otimes Y') \times A(X,Y)
	\arrow[d,"\id \times \st_{Y'}"]\\
	A(X \otimes X',Y \otimes X') \times B(Y \otimes X', X \otimes Y')
	\arrow[d,"\ell",swap]
	&& B(Y \otimes X', X \otimes Y') \times A(X \otimes Y',Y \otimes Y')
	\arrow[d,"r"] \\
	B(X \otimes X', X \otimes Y') \arrow[dr,"\cst'_X",swap]
	&& B(Y \otimes X', Y \otimes Y') \arrow[dl,"\cst'_Y"]\\
	& B(X',Y') &
\end{tikzcd}
\end{center}
\end{definition}

%


\end{document}